\documentclass[oribibl]{llncs}
\usepackage{llncsdoc}
\usepackage[english]{babel}
\usepackage{verbatim}

\usepackage{psfrag}
\usepackage{graphicx}
\usepackage{textcomp}
\usepackage{xcolor}

\usepackage{amsmath}
\usepackage{amsmath}
\usepackage{amsfonts}
\usepackage{amssymb}
\usepackage{mathrsfs}
\usepackage{multirow}

\usepackage{pstricks}
\usepackage{pifont}
\usepackage{pxfonts}
\usepackage{latexsym}
\usepackage{booktabs}
\usepackage{array}
\usepackage{color}
\usepackage{colortbl}

\newcommand{\st}{:}
\newcommand{\MK}{\mathsf{MK}}
\newcommand{\PK}{\mathsf{PK}}
\newcommand{\SK}{\mathsf{SK}}
\newcommand{\CT}{\mathsf{CT}}
\newcommand{\Dec}{\mathsf{Decrypt}}
\newcommand{\Enc}{\mathsf{Encrypt}}
\newcommand{\Key}{\mathsf{KeyGen}}
\newcommand{\Set}{\mathsf{Setup}}

\newcommand{\A}{\mathcal{A}}
\newcommand{\B}{\mathcal{B}}
\newcommand{\E}{\mathcal{E}}

\begin{document}


\title{Key-Policy Multi-Authority Attribute-Based Encryption}
\author{Riccardo Longo\thanks{\tt riccardolongomath@gmail.com, \it Department of Mathematics, University of Trento},\, Chiara Marcolla\thanks{ \tt chiara.marcolla@gmail.com, \it Department of Mathematics, University of Turin},\, Massimiliano Sala\thanks{\tt maxsalacodes@gmail.com, \it Department of Mathematics, University of Trento}}
\institute{}
\maketitle


\begin{abstract}
Bilinear groups are often used to create Attribute-Based Encryption (ABE) algorithms.
In particular, they have been used to create an ABE system with multi authorities, but limited to the ciphertext-policy instance.
Here, for the first time, we propose a multi-authority key-policy ABE system.\\
In our proposal, the authorities may be set up in any moment and without any coordination.
A party can simply act as an ABE authority by creating its own public parameters and issuing private keys to the users.
A user can thus encrypt data choosing both a set of attributes and a set of trusted authorities, maintaining full control unless all his chosen authorities collude against him.\\
We prove our system secure under the bilinear Diffie-Hellman assumption.
\end{abstract}





\section{Introduction}
\label{intro}
  The key feature that makes the cloud so attracting nowadays is the great accessibility it provides: users can access their data through the Internet from anywhere.
  Unfortunately, at the moment the protection offered for sensitive information is questionable and access control is one of the greatest concerns.
  Illegal access may come from outsiders or even from insiders without proper clearance.
  One possible approach for this problem is to use Attribute-Based Encryption (ABE) that provides cryptographically enhanced access control functionality in encrypted data.\\

    ABE developed from Identity Based Encryption, a scheme proposed by Shamir \cite{shamir1985identity} in 1985 with the first constructions obtained in 2001 by Boneh and Franklin \cite{boneh2001identity}.
    The use of bilinear groups, in particular the Tate and Weil pairings on elliptic curves \cite{boneh2001identity}, was the winning strategy that finally allowed to build  schemes following the seminal idea of Shamir.
    Bilinear groups came in nicely when a preliminary version of ABE was invented by Sahai and Waters \cite{sahai2005fuzzy} in 2005.
    Immediately afterwards, Goyal, Pandey, Sahai, and Waters  \cite{goyal2006attribute} formulated the two complimentary forms of ABE
    which are nowadays standard: ciphertext-policy ABE  and key-policy ABE.
    In a ciphertext-policy ABE system, keys are associated with sets of attributes and ciphertexts are associated with access policies.
    In a KP-ABE system, the situation is reversed: keys are associated with access policies and ciphertexts are associated with sets of attributes.
    Several developments in efficiency and generalizations have been obtained
    for key-policy ABE, e.g. \cite{attrapadung2012attribute}, \cite{hohenberger2013attribute}, \cite{ostrovsky2007attribute}.
    A first implementation of ciphertext-policy ABE has been achieved by Bethencourt et al. \cite{bethencourt2007ciphertext} in 2007 but the proofs of security of the ciphertext-policy ABE remained unsatisfactory since they were based on an assumption independent of the algebraic structure of the group (the generic group model).
    It is only with the work of Waters \cite{waters2011ciphertext} that the first non-restricted ciphertext-policy ABE scheme was built with a security dependent on variations of the DH assumption on bilinear groups.
    Related to the work we propose in this paper is the construction for multiple authorities (ciphertext-policy ABE) that have been proposed in \cite{chase2007multi}, \cite{chase2009improving} and \cite{lewko2011decentralizing}. \\
    However, before the present paper no multi-authority KP-ABE scheme has appeared in the literature with a proof of security.

\paragraph*{Our construction}
  In this paper we present the first multi-authority KP-ABE scheme.
  In our system, after the creation of an initial set of common parameters, the authorities may be set up in any moment and without any coordination.
  A party can simply act as an ABE authority by creating a public parameters and issuing private keys to different users (assigning access policies while doing so).
  A user can encrypt data under any set of attributes specifying also a set of \emph{trusted} authorities, so the encryptor maintains high control.
  Also, the system does not require any central authority.
  Our scheme has both very short single-authority keys, that compensate the need of multiple keys (one for authority), and  also very short ciphertexts. Moreover, the pairing computations in the bilinear group are involved only during the decryption phase, obtaining this way significant advantages in terms of encryption times.

  Even if the authorities are collaborating, the existence of just one non-cheating authority guarantees that no illegitimate party (including authorities) has access to the encrypted data.

  We prove our scheme secure using the classical bilinear Diffie-Hellman assumption.

  \paragraph*{Organization}
  This paper is organized as follows.
  In Section~\ref{lsss} we present the main mathematical tools used in the construction of multi authority KP-ABE scheme.
  In Section~\ref{schema} we  explain in detail our multi authority KP-ABE scheme and its security is proven under standard, non-interactive assumptions in the selective set model.
  Finally conclusions are drawn in Section \ref{conc}.


\section{Preliminaries}
\label{lsss}
We do not prove original results here, we only provide what we need for our construction.
See the cited references for more details on these arguments.\\

Let $\mathbb{G}_1, \mathbb{G}_2$ be groups of the same prime order $p$.
    \begin{definition}[Pairing]
      A symmetric pairing is a bilinear map $e$ such that $e: \mathbb{G}_1 \times \mathbb{G}_1 \rightarrow \mathbb{G}_2$ has the following properties:
      \begin{itemize}
        \item Bilinearity: $\forall g, h \in \mathbb{G}_1, \forall a, b \in \mathbb{Z}_p, \quad e(g^a, h^b) = e(g, h)^{ab}$.
        \item Non-degeneracy: for $g$ generator of $\mathbb{G}_1$, $e(g, g)\neq 1$.
      \end{itemize}
    \end{definition}

    \begin{definition}[Bilinear Group]
      $\mathbb{G}_1$ is a Bilinear group if the conditions above hold and both the group operations in $\mathbb{G}_1$ and $\mathbb{G}_2$ as well as the bilinear map $e$ are efficiently computable.
    \end{definition}

  Let $a, b, s, z \in \mathbb{Z}_p$ be chosen at random and $g$ be a generator of the bilinear group $\mathbb{G}_1$.
  The decisional bilinear Diffie-Hellman (BDH) problem consists in constructing an algorithm $\B(A = g^a, B=g^b, S=g^s, T) \rightarrow \{0,1\}$ to efficiently distinguish between the tuples $(A,B,S,e(g,g)^{abs})$ and $(A,B,S,e(g,g)^{z})$ outputting respectively 1 and 0.
  The advantage of $\mathcal{B}$ is:
  $$
    Adv_{\B} = \Big|\text{Pr}\left[\B(A, B, S, e(g, g)^{abs}) = 1\right] - \text{Pr}\left[\B(A, B, S, e(g, g)^z )= 1\right]\Big|
  $$

\noindent where the probability is taken over the random choice of the generator $g$, of $a, b, s, z$ in $\mathbb{Z}_p$, and the random bits possibly consumed by $\B$ to compute the response.

    \begin{definition}[BDH Assumption] \label{BDH}
      The decisional BDH assumption holds if no probabilistic polynomial-time algorithm $\B$ has a non-negligible advantage in solving the decisional BDH problem.
    \end{definition}

 Access structures define who may and who may not access the data, giving the sets of attributes that have clearance.

\begin{definition}[Access Structure]
    An access structure $\mathbb{A}$ on a universe of attributes $U$ is the set of the subsets $S \subseteq U$ that are authorized.
    That is, a set of attributes $S$ satisfies the policy described by the access structure $\mathbb{A}$ if and only if $S \in \mathbb{A}$.
\end{definition}

  They are used to describe a policy of access, that is the rules that prescribe who may access to the information.
  If these rules are constructed using only \textsc{AND, OR} and \textsc{threshold} operators on the attributes, then the access structure is \emph{monotonic}.

  \begin{definition}[Monotonic Access Structure]
    \label{mono}
    An access structure $\mathbb{A}$ is said to be monotonic if given $S_0 \subseteq S_1 \subseteq U$ it holds
$$
      S_0 \in \mathbb{A} \Longrightarrow S_1 \in \mathbb{A}
$$
  \end{definition}

  An interesting property is that monotonic access structures (i.e. access structures $\mathbb{A}$ such that if $S$ is an authorized set and $S \subseteq S'$ then also $S'$ is an authorized set) may be associated to linear secret sharing schemes (LSSS).
  In this setting the parties of the LSSS are the attributes of the access structure.

  A LSSS may be defined as follows (adapted from \cite{beimel1996secure}).
  \begin{definition}[Linear Secret-Sharing Schemes (LSSS)]
    \label{lsssdef}
    A secret-sharing scheme $\Pi$ over a set of parties $P$ is called linear (over $\mathbb{Z}_p$) if
    \begin{enumerate}
      \item The shares for each party form a vector over $\mathbb{Z}_p$.
      \item
        There exists a matrix $M$ with $l$ rows and $n$ columns called the share-generating matrix for $\Pi$.
        For all $i \in \{1,\ldots,l\}$ the $i$-th row of $M$ is labeled via a function $\rho$, that associates $M_i$ to the party $\rho(i)$.
        Considering the vector $\vec{v} = (s, r_2, \ldots, r_n) \in \mathbb{Z}_p^n$, where $s \in \mathbb{Z}_p$ is the secret to be shared, and $r_i \in \mathbb{Z}_p$, with $i \in \{2,\ldots,n\}$ are randomly chosen, then $M \vec{v}$ is the vector of $l$ shares of the secret $s$ according to $\Pi$.
        The share $(M\vec{v})_i = M_i \vec{v}$ belongs to party $\rho(i)$.
    \end{enumerate}
  \end{definition}

  It is shown in \cite{beimel1996secure} that every linear secret sharing-scheme according to the above definition also enjoys the linear reconstruction property, defined as follows: suppose that $\Pi$ is an LSSS for the access structure $\mathbb{A}$. Let $S \in  \mathbb{A}$ be any authorized set, and let $I \subseteq \{1,\ldots,l\}$ be defined as $I = \{i \st \rho(i) \in  S\}$.
  Then, there exist constants $w_i \in  \mathbb{Z}_p$, with $i\in I$ such that, if ${\lambda_i }$ are valid shares of any secret $s$ according to $\Pi$, then
  \begin{equation}
    \sum_{i\in I} w_i \lambda_i = s
  \end{equation}

  Furthermore, it is shown in \cite{beimel1996secure} that these constants $w_i$ can be found in time polynomial in the size of the share-generating matrix $M$.

  Note that the vector $(1, 0, \ldots , 0)$ is the “target” vector for the linear secret sharing scheme.
  Then, for any set of rows $I$ in $M$, the target vector is in the span of $I$ if and only if $I$ is an authorized set.
  This means that if $I$ is not authorized, then for any choice of $c \in \mathbb{Z}_p$ there will exist a vector $\vec{u}$ such that $u_1 = c$ and
$$
    M i \cdot \vec{w} = 0 \qquad \forall i \in I
$$

  In the first ABE schemes the access formulas are typically described in terms of access trees.
  The appendix of \cite{lewko2011decentralizing} is suggested for a discussion of how to perform a conversion from access trees to LSSS.

  See \cite{goyal2006attribute}, \cite{beimel1996secure} and \cite{liu2010efficiently} for more details about LSSS and access structures.


\section{Our Construction}
\label{schema}

This section is divided in three parts. We start with  definitions of Multi-Authority Key-Policy ABE and of CPA selective security. In the second part we present in detail our first scheme and, finally, we prove the security of this scheme under the classical BDH assumption in the selective set model.

A security parameter will be used to determine the size of the bilinear group used in the construction, this parameter represents the order of complexity of the assumption that provides the security of the scheme.
Namely, first the complexity is chosen thus fixing the security parameter, then this value is used to compute the order that the bilinear group must have in order to guarantee the desired complexity, and finally a suitable group is picked and used.

\subsection{Multi Authority KP-ABE Structure and Security}\label{KP-ABE}

  In this scheme, after the common universe of attributes and bilinear group are agreed, the authorities set up independently their master key and public parameters.
  The master key is subsequently used to generate the private keys requested by users.
  Users ask an authority for keys that embed a specific access structure, and the authority issues the key only if it judges that the access structure suits the user that requested it.
  Equivalently an authority evaluates a user that requests a key, assigns an access structure, and gives to the user a key that embeds it.
  When someone wants to encrypt, it chooses a set of attributes that describes the message (and thus determines which access structures may read it) and a set of trusted authorities.
  The ciphertext is computed using the public parameters of the chosen authorities, and may be decrypted only using a valid key for each of these authorities.
  A key with embedded access structure $\mathbb{A}$ may be used to decrypt a ciphertext that specifies a set of attributes $S$ if and only if $S \in \mathbb{A}$, that is the structure considers the set authorized.

  This scheme is secure under the classical BDH assumption in the selective set model, in terms of chosen-ciphertext indistinguishability.

  The security game is formally defined as follows.\\

  \noindent Let $\mathcal{E}=(\Set,\Enc,\Key,\Dec)$ be a MA-KP-ABE scheme for a message space $\mathcal M$, a universe of authorities $X$ and an access structure space $\mathcal{G}$ and consider the following MA-KP-ABE experiment $\mathsf{MA}$-$\mathsf{KP}$-$\mathsf{ABE}$-$\mathsf{Exp}_{\A,\E}(\lambda,U)$  for an adversary $\mathcal A$, parameter $\lambda$ and attribute universe $U$:
 \begin{description}
      \item [Init.]
        The adversary declares the set of attributes $S$ and the set of authorities $A \subseteq X$ that it wishes to be challenged upon. Moreover it selects the \textit{honest authority} $k_0 \in A$.
      \item [Setup.]
        The challenger runs the Setup algorithm, initializes the authorities and gives to the adversary the public parameters.
      \item [Phase I.]
        The adversary issues queries for private keys of any authority, but $k_0$ answers only to queries for keys for access structures $\mathbb{A}$ such that $S \notin \mathbb{A}$.
        On the contrary the other authorities respond to every query.
      \item [Challenge.]
        The adversary submits two equal length messages $m_0$ and $m_1$.
        The challenger flips a random coin $b \in \{0, 1\}$, and encrypts $m_b$ with $S$ for the set of authorities $A$.
        The ciphertext is passed to the adversary.
      \item [Phase II.]
        Phase \texttt{I} is repeated.
      \item [Guess.]
        The adversary outputs a guess $b'$ of $b$.
  \end{description}

  \begin{definition}[MA-KP-ABE Selective Security]
  The MA-KP-ABE scheme $\E$ is  CPA selective secure (or secure against   chosen-plaintext attacks) for attribute universe $U$ if for all probabilistic polynomial-time adversaries $\mathcal A$, there exists a negligible function $negl$ such that:
  $$\Pr[\mathsf{MA}\mbox{-}\mathsf{KP}\mbox{-}\mathsf{ABE}\mbox{-}\mathsf{Exp}_{\A,\E}(\lambda,U) = 1] \leq \frac{1}{2} + negl(\lambda).$$
  \end{definition}

\subsection{The Scheme}
  The scheme plans a set $X$ of independent authorities, each with their own pa\-ra\-me\-ters, and it sets up an encryption algorithm that lets the encryptor choose a set $A \subseteq X$ of authorities, and combines the public parameters of these in such a way that an authorized key for each authority in $A$ is required to successfully decrypt.\\
  Our scheme consists of three randomized algorithms ($\Set, \Key,$ $\Enc$) plus the decryption $\Dec$.
  The techniques used are inspired from the scheme of Goyal et al. in \cite{goyal2006attribute}.
  The scheme works in a bilinear group $\mathbb{G}_1$ of prime order $p$, and uses LSSS matrices to share secrets according to the various access structures.
  Attributes are seen as elements of $\mathbb{Z}_p$.

  The description of the algorithms follows.

  \paragraph*{$\Set (U, g, \mathbb{G}_1) \rightarrow (\PK_k, \MK_k).$}
    Given the universe of attributes $U$ and a generator $g$ of $\mathbb{G}_1$ each authority sets up independently its parameters.
    For $k \in X$ the Authority $k$ chooses uniformly at random $\alpha_k \in \mathbb{Z}_p$, and $z_{k, i} \in \mathbb{Z}_p$ for each $i \in U$.
    Then the public parameters $\PK_k$ and the master key $\MK_k$ are:
    $$
      \PK_k = \left(e(g,g)^{\alpha_k}, \{g^{z_{k, i}}\}_{i \in U} \right) \qquad \MK_k = \left(\alpha_k, \{z_{k, i}\}_{i \in U}\} \right)
    $$

  \paragraph*{$\Key_k(\MK_k, (M_k, \rho_k)) \rightarrow \SK_k.$}
    The key generation algorithm for the authority $k$ takes as input the master secret key $\MK_k$ and an LSSS access structure $(M_k, \rho_k)$, where $M_k$ is an $l\times n$ matrix on $\mathbb{Z}_p$ and $\rho_k$ is a function which associates rows of $M_k$ to attributes.
    It chooses uniformly at random a vector $\vec{v}_k \in \mathbb{Z}_p^n$ such that $v_{k,1} = \alpha_k$.
    Then it computes the shares $\lambda_{k,i} = M_{k,i} \vec{v}_k$ for $1 \leq i \leq l$ where $M_{k, i}$ is the $i$-th row of $M_k$.
    Then the private key $\SK_k$ is:
    $$
      \SK_k = \left\{K_{k, i} = g^{\frac{\lambda_{k,i}}{z_{k, \rho_{k}(i)}}}\right\}_{1\leq i \leq l}
    $$

  \paragraph*{$\Enc(m, S, \{\PK_k\}_{k \in A}) \rightarrow \CT.$}
    The encryption algorithm takes as input the public parameters, a set $S$ of attributes and a message $m$ to encrypt.
    It chooses $s \in \mathbb{Z}_p$ uniformly at random and then computes the ciphertext as:
    $$
      \CT = \left(S, C' = m \cdot \left(\prod_{k\in A} e(g,g)^{\alpha_k} \right)^s, \{C_{k,i} = (g^{z_{k, i}})^s\}_{k\in A, ~i \in S} \right)
    $$

  \paragraph*{$\Dec(\CT, \{\SK_k\}_{k\in A})  \rightarrow m'.$}
    The input is a ciphertext for a set of attributes $S$ and a set of authorities $A$ and an authorized key for every authority cited by the ciphertext.
    Let $(M_k, \rho_k)$ be the LSSS associated to the key $k$, and suppose that $S$ is authorized for each $k \in A$.
    The algorithm for each $k\in A$ finds $w_{k,i} \in \mathbb{Z}_p, i \in I_k$ such that
    \begin{equation}
      \label{alfak}
      \sum_{i \in I_k} \lambda_{k, i} w_{k, i} = \alpha_k
    \end{equation}
    for appropriate subsets $I_k \subseteq S$ and then proceeds to reconstruct the original message computing:
    \begin{align*}
      m' & = \frac{C'}{\prod_{k \in A} \prod_{i \in I_k} e(K_{k, i}, C_{k, \rho_{k}(i)})^{w_{k,i}}} \\
      & = \frac{m \cdot \left(\prod_{k \in A} e(g,g)^{\alpha_k} \right)^s}{\prod_{k \in A} \prod_{i \in I_k} e\left(g^{\frac{\lambda_{k,i}}{z_{k, \rho_{k}(i)}}}, (g^{z_{k, \rho_{k}(i)}})^s\right)^{w_{k,i}}} \nonumber \\
      & = \frac{m \cdot e(g, g)^{s(\sum_{k\in A}\alpha_k)}}{\prod_{k \in A} e(g, g)^{s \sum_{i\in I_k} w_{k, i} \lambda_{k, i}}} \nonumber \\
      & \stackrel{*}{=} \frac{m \cdot e(g, g)^{s(\sum_{k\in A}\alpha_k)}}{e(g, g)^{s (\sum_{k\in A}\alpha_k)}} = m
    \end{align*}
    Where $\stackrel{*}{=}$ follows from property (\ref{alfak}).

\subsection{Security}
  The scheme is proved secure under the BDH assumption (Definition \ref{BDH}) in a selective set security game in which every authority but one is supposed curious (or corrupted or breached) and then it will issue even keys that have enough clearance for the target set of attributes, while the honest one issues only unauthorized keys.
  Thus if at least one authority remains trustworthy the scheme is secure.\\
  The security  is provided by the following theorem.

  \begin{theorem}
    If an adversary can break the scheme, then a simulator can be constructed to play the Decisional BDH game with a non-negligible advantage.
  \end{theorem}

  \begin{proof}
    Suppose there exists a polynomial-time adversary $\mathcal{A}$, that can attack the scheme in the Selective-Set model with advantage $\epsilon$.
    Then a simulator $\mathcal{B}$ can be built that can play the Decisional BDH game with advantage $\epsilon/2$.
    The simulation proceeds as follows.

    \paragraph*{Init}
      The simulator takes in a BDH challenge $g, g^a, g^b, g^s, T$.
      The adversary gives the algorithm the challenge access structure $S$.

    \paragraph*{Setup}
      The simulator chooses random $r_k \in \mathbb{Z}_p$ for $k \in A \setminus \{k_0\}$ and implicitly sets $\alpha_k = -r_k b$ for $k \in A \setminus \{k_0\}$ and $\alpha_{k_0} = a b + b \sum_{k \in A \setminus \{k_0\}} r_k$ by computing:
      \begin{align*}
        e(g,g)^{\alpha_{k_0}} &= e(g^a, g^b) \prod_{k \in A \setminus \{k_0\}} (g^b, g^{r_k}) \\
        e(g,g)^{\alpha_k} &= e(g^b, g^{-r_k}) \qquad \forall k \in A \setminus \{k_0\}
      \end{align*}
      Then it chooses $z_{k, i}' \in \mathbb{Z}_p$ uniformly at random for each $i \in U$, $k \in A$ and implicitly sets
      $$
        z_{k, i} =
        \begin{cases}
          z_{k, i}' &\text{ if } i \in S \\
          b z_{k, i}' &\text{ if } i \notin S
        \end{cases}
      $$
      Then it can publish the public parameters computing the remaining values as:
      $$
        g^{z_{k, i}} =
        \begin{cases}
          g^{z_{k, i}'} &\text{ if } i \in S \\
          (g^b)^{z_{k, i}'} &\text{ if } i \notin S
        \end{cases}
      $$

    \paragraph*{Phase I}
      In this phase the simulator answers private key queries.
      For the queries made to the authority $k_0$ the simulator has to compute the $K_{k_0,i}$ values of a key for an access structure $(M,\rho)$ with dimension $l \times n$ that is not satisfied by $S$.
      Therefore for the properties of an LSSS it can find a vector $\vec{y} \in \mathbb{Z}_p^n$ with $y_1 = 1$ fixed such that
      \begin{equation}\label{y-sempl}
        M_i \vec{y} = 0 \qquad \forall i \mbox{ such that } \rho(i) \in S
      \end{equation}
      Then it chooses uniformly at random a vector $\vec{v} \in \mathbb{Z}_p^n$ and implicitly sets the shares of $\alpha_{k_0} = b(a+ \sum_{k \in A \setminus \{k_0\}} r_k)$ as
      $$
        \lambda_{k_0,i} = b \sum_{j=1}^n M_{i,j} (v_j + (a + \sum_{k \in A \setminus \{k_0\}} r_k - v_1)y_j)
      $$
      Note that $\lambda_{k_0,i} = \sum_{j=1}^n M_{i,j} u_j$ where $u_j = b (v_j + (a + \sum_{k \in A \setminus \{k_0\}} r_k - v_1)y_j)$ thus \linebreak[4] ${u_1 = b (v_1 + (a + \sum_{k \in A \setminus \{k_0\}} r_k - v_1) 1) = ab + b \sum_{k \in A \setminus \{k_0\}} r_k= \alpha_{k_0} }$ so the shares are valid.
      Note also that from (\ref{y-sempl}) it follows that
      $$
        \lambda_{k_0,i} = b \sum_{j=1}^n M_{i,j}v_j \quad \forall i \mbox{ such that } \rho(i) \in S
      $$
      Thus if $i \mbox{ is such that } \rho(i) \in S$ the simulator can compute
      $$
        K_{k_0,i} = (g^b)^{\frac{\sum_{j=1}^n M_{i,j}v_j}{z_{k_0,\rho(i)}'}} = g^{\frac{\lambda_{k_0,i}}{z_{k_0, \rho(i)}}}
      $$
      Otherwise, if $i \mbox{ is such that } \rho(i) \notin S$ the simulator computes
      \begin{align*}
        K_{k_0,i} &= g^{\frac{\sum_{j=1}^n M_{i,j}(v_j +(r - v_1)y_j)}{z_{k_0,\rho(i)}'}} (g^a)^{\frac{\sum_{j=1}^n M_{i,j}y_j}{z_{k_0,\rho(i)}'}}  = g^{\frac{\lambda_{1,i}}{z_{k_0, \rho(i)}}}
      \end{align*}
      Remembering that in this case $z_{k_0, \rho(i)} := b z_{k_0, \rho(i)}'$.
      Finally for the queries to the other authorities $k \in A \setminus \{k_0\}$, the simulator chooses uniformly at random a vector $\vec{t}_k \in \mathbb{Z}_p^n$ such that $t_{k, 1} = -r_k$ and implicitly sets the shares $\lambda_{k,i} = b \sum_{j=1}^n M_{i,j} t_{k,j}$ by computing
      $$
        K_{k,i} =
        \begin{cases}
          (g^b)^{\frac{\sum_{j=1}^n M_{i,j} t_{k,j}}{z_{k, \rho(i)}'}} = g^{\frac{b \sum_{j=1}^n M_{i,j} t_{k,j}}{z_{k, \rho(i)}'}} = g^{\frac{\lambda_{k,i}}{z_{k, \rho(i)}}} &\text{ if } i \in S \\
          g^{\frac{\sum_{j=1}^n M_{i,j} t_{k,j}}{z_{k, \rho(i)}'}} = g^{\frac{b \sum_{j=1}^n M_{i,j} t_{k,j}}{b z_{k, \rho(i)}'}} = g^{\frac{\lambda_{k,i}}{z_{k, \rho(i)}}} &\text{ if } i \notin S
        \end{cases}
      $$

    \paragraph*{Challenge}
      The adversary gives two messages $m_0, m_1$ to the simulator.
      It flips a coin $\mu$.
      It creates:
      \begin{align*}
        C' &= m_{\mu} \cdot T \stackrel{*}{=} m_{\mu} \cdot e(g,g)^{a b s} \\
        &=  m_{\mu} \cdot \left(e(g,g)^{(a b + b\left(\sum_{k \in A \setminus \{k_0\}} r_k\right)}\prod_{k \in A \setminus \{k_0\}} e(g,g)^{b r_k} \right)^s \nonumber \\
        C_{k,i} &= (g^s)^{z_{k, \rho(i)}'} = g^{s z_{k, \rho(i)}} \qquad k \in A, \quad i \in S
      \end{align*}
      Where the equality $\stackrel{*}{=}$ holds if and only if the BDH challenge was a valid tuple (i.e. $T$ is non-random).

    \paragraph*{Phase II}
      During this phase the simulator acts exactly as in \emph{Phase~I}.

    \paragraph*{Guess}
      The adversary will eventually output a guess $\mu'$ of $\mu$.
      The simulator then outputs $0$ to guess that $T = e(g, g)^{a b s}$ if $\mu' = \mu$; otherwise, it outputs $1$ to indicate that it believes $T$ is a random group element in $\mathbb{G}_2$.
      In fact when $T$ is not random the simulator $\mathcal{B}$ gives a perfect simulation so it holds:
      $$
        Pr\left[\mathcal{B}\left(\vec{y},T=e(g,g)^{a b s}\right)=0\right] = \frac{1}{2} + \epsilon
      $$
      On the contrary when $T$ is a random element $R \in \mathbb{G}_2$ the message $m_{\mu}$ is completely hidden from the adversary point of view, so:
      $$
        Pr\left[\mathcal{B}\left(\vec{y},T=R\right)=0\right] = \frac{1}{2}
      $$
      Therefore, $\mathcal{B}$ can play the decisional BDH game with non-negligible advantage$~\frac{\epsilon}{2}$.
  \end{proof}


\section{Related Works and Final Comments}
\label{conc}

  Our scheme gives a solution addressing the problem of faith in the authority, specifically the concerns arisen by key escrow and clearance check.
  Key escrow is a setting in which a party (in this case the authority) may obtain access to private keys and thus it can decrypt any ciphertext.
  Normally the users have faith in the authority and assume that it will not abuse its powers.
  The problem arises when the application does not plan a predominant role and there are trust issues selecting any third party that should manage the keys.
  In this situation the authority is seen as \emph{honest but curious}, in the sense that it will provide correct keys to users (then it is not malicious) but will also try to access to data beyond its competence.
  It is clear that as long as a single authority is the unique responsible to issue the keys, there is no way to prevent key escrow.
  Thus the need for multi-authority schemes arises.

  The second problem is more specific for KP-ABE.
  In fact, the authority has to assign to each user an appropriate access structure that represents what the user can and cannot decrypt.
  Therefore, the authority has to be trusted not only to give correct keys and to not violate the privacy, but also to perform correct checks of the users' clearances and to assign correct access structures accordingly.
  Therefore, in addition to satisfying the requirements of not being \emph{malicious} and not being \emph{curious}, the authority must also not have been \emph{breached}, in the sense that a user's keys must embed access structures that faithfully represent that user's level of clearance, and that no one has access to keys with a higher level of clearance than the one they are due.
  In this case, to add multiple authorities to the scheme gives to the encryptor the opportunity to request more guarantees about the legitimacy of the decryptor's clearance.
  In fact, each authority checks the users independently, so the idea  is to request that the decryption proceeds successfully only when a key for each authority of a given set $A$ is used.
  This means that the identity of the user has been checked by every selected authority, and the choice of these by the encryptor models the trust that he has in them.
  Note that if these authorities set up their parameters independently and during encryption these parameters are bound together irrevocably, then no authority can single-handedly decrypt any ciphertext and thus key escrow is removed.
  So our KP-ABE schemes guarantee a protection against both breaches and curiosity.

  The scheme proposed has very short single-keys (just one element per row of the access matrix) that compensates for the need of multiple single-keys (one for cited authority) in the decryption.
  Ciphertexts are also very short (the number of elements is linear in the number of authorities times the number of attributes under which it has been encrypted) thus the scheme is efficient under this aspect.
  Moreover, there are \emph{no} pairing computations involved during encryption and this means significant advantages in terms of encryption times.
  Decryption time is not constant in the number of pairings (e.g. as in the scheme presented in \cite{hohenberger2013attribute} or the one in \cite{waters2011ciphertext}) but requires $\sum_{k\in A} l_k$ pairings where $A$ is the set of authorities involved in encryption and $l_k$ is the number of rows of the access matrix of the key given by authority $k$, so to maintain the efficiency of the scheme only a few authorities should be requested by the encryptor.

  Taking a more historical perspective, the problem of multi-authority ABE is not novel and a few solution have been proposed.
  The problem of building ABE systems with multiple authorities was proposed by Sahai and Waters. This problem with the presence of a central authority was firstly considered by Chase \cite{chase2007multi}  and then improved by Chase and Chow \cite{chase2009improving}, constructing simple-threshold schemes in the case where attributes are divided in disjoint sets, each controlled by a different authority.
  These schemes are also shown to be extensible from simple threshold to KP-ABE, but retaining the partition of attributes and requiring the involvement of every authority in the decryption.
  In those works the main goal is to relieve the central authority of the burden of generating key material for every user and add resiliency to the system.
  Multiple authorities manage the attributes, so that each has less work and the whole system does not get stuck if one is down.
  Another approach has been made by Lin et al. \cite{lin2010secure} where a central authority is not needed but a parameter directly sets the efficiency and number of users of the scheme.

  More interesting results have been achieved for CP schemes, in which the partition of the attributes makes more sense, for example \cite{muller2009distributed}.
  The most recent and interesting result may be found in \cite{lewko2011decentralizing}, where Lewko and Waters propose a scheme where is not needed a central authority or coordination between the authorities, each controlling disjoint sets of attributes.
  They used composite bilinear groups and via Dual System Encryption (introduced by Waters \cite{waters2009dual} with techniques developed with Lewko \cite{lewko2010new}) proved their scheme fully secure following the example of Lewko et al. \cite{lewko2010fully}.
  They allow the adversary to statically corrupt authorities choosing also their master key.
  Note however that they did not specifically address key escrow but distributed workload.

  Our results of this article retain relevance since they address a different setting.
  In fact, with this extensions the differences in the situations of ciphertext-policy ABE and KP-ABE model become more distinct.
  For example a situation that suits the scheme proposed here, but not the one of Lewko and Waters is the following.
  Consider company branches dislocated on various parts of the world, each checking its personnel and giving to each an access policy (thus acting as authorities).
  This scheme allows encryptions that may be decrypted by the manager of the branch (simply use only one authority as in classic ABE) but also more secure encryptions that require the identity of the decryptor to be guaranteed by more centers, basing the requirements on which branches are still secure and/or where a user may actually authenticate itself.

  Moreover, we observe that although the scheme of \cite{lewko2011decentralizing} is proven fully secure (against selective security), the construction is made in composite bilinear groups.
  It is in fact compulsory when using Dual System encryption, but this has drawbacks in terms of group size (integer factorization has to be avoided) and the computations of pairings and group operations are less efficient.
  This fact leads to an alternative construction in prime order groups in the same paper, that however is proven secure only in the generic group and random oracle model.
  These considerations demonstrate that our construction in prime groups under basic assumptions retain validity and interest.

\paragraph*{Aknowledgements}
  Most results in this paper are contained in the first's author Msc. thesis \cite{tesi5-longo} who wants to thank his supervisors: the other two authors.

\bibliographystyle{splncs}
\bibliography{bibtexABE}

\end{document}